\documentclass[11pt,a4paper]{amsart}
\usepackage{amssymb, amstext, amscd, amsmath, color}
\usepackage{graphicx}

\usepackage{url}

\usepackage{tikz,enumerate}
\usepackage{pgfplots}
\usepackage{url}
\usepackage{kbordermatrix} 
\renewcommand{\kbldelim}{[} 
\renewcommand{\kbrdelim}{]}

\usepackage{hhline}

\pgfplotsset{compat=1.10}
\begin{document}

\newtheorem{thm}{Theorem}[section]
\newtheorem{cor}[thm]{Corollary}
\newtheorem{prop}[thm]{Proposition}
\newtheorem{lem}[thm]{Lemma}
%
\theoremstyle{definition}
\newtheorem{rem}[thm]{Remark}
\newtheorem{defn}[thm]{Definition}
\newtheorem{note}[thm]{Note}
\newtheorem{eg}[thm]{Example}
\newcommand{\Prf}{\noindent\textbf{Proof.\ }}
\newcommand{\bx}{\hfill$\blacksquare$\medbreak}
\newcommand{\upbx}{\vspace{-2.5\baselineskip}\newline\hbox{}%
\hfill$\blacksquare$\newline\medbreak}
\newcommand{\eqbx}[1]{\medbreak\hfill\(\displaystyle #1\)\bx}

\newcommand{\FFock}{\mathcal{F}}
\newcommand{\kil}{\mathsf{k}}
\newcommand{\Hil}{\mathsf{H}}
\newcommand{\hil}{\mathsf{h}}
\newcommand{\Kil}{\mathsf{K}}
\newcommand{\Real}{\mathbb{R}}
\newcommand{\Rplus}{\Real_+}

%

\newcommand{\bC}{{\mathbb{C}}}
\newcommand{\bD}{{\mathbb{D}}}
\newcommand{\bK}{{\mathbb{K}}}
\newcommand{\bN}{{\mathbb{N}}}
\newcommand{\bQ}{{\mathbb{Q}}}
\newcommand{\bR}{{\mathbb{R}}}
\newcommand{\bT}{{\mathbb{T}}}
\newcommand{\bX}{{\mathbb{X}}}
\newcommand{\bZ}{{\mathbb{Z}}}
\newcommand{\bH}{{\mathbb{H}}}
\newcommand{\BH}{{\B(\H)}}
\newcommand{\bsl}{\setminus}
\newcommand{\ca}{\mathrm{C}^*}
\newcommand{\cstar}{\mathrm{C}^*}
\newcommand{\cenv}{\mathrm{C}^*_{\text{env}}}
\newcommand{\rip}{\rangle}
\newcommand{\ol}{\overline}
\newcommand{\td}{Widetilde}
\newcommand{\Wh}{Widehat}
\newcommand{\sot}{\textsc{sot}}
\newcommand{\Wot}{\textsc{wot}}
\newcommand{\Wotclos}[1]{\ol{#1}^{\textsc{wot}}}
 \newcommand{\A}{{\mathcal{A}}}
 \newcommand{\B}{{\mathcal{B}}}
 \newcommand{\C}{{\mathcal{C}}}
 \newcommand{\D}{{\mathcal{D}}}
 \newcommand{\E}{{\mathcal{E}}}
 \newcommand{\F}{{\mathcal{F}}}
 \newcommand{\G}{{\mathcal{G}}}
\renewcommand{\H}{{\mathcal{H}}}
 \newcommand{\I}{{\mathcal{I}}}
 \newcommand{\J}{{\mathcal{J}}}
 \newcommand{\K}{{\mathcal{K}}}
\renewcommand{\L}{{\mathcal{L}}}
 \newcommand{\M}{{\mathcal{M}}}
 \newcommand{\N}{{\mathcal{N}}}
\renewcommand{\O}{{\mathcal{O}}}
\renewcommand{\P}{{\mathcal{P}}}
 \newcommand{\Q}{{\mathcal{Q}}}
 \newcommand{\R}{{\mathcal{R}}}
\renewcommand{\S}{{\mathcal{S}}}
 \newcommand{\T}{{\mathcal{T}}}
 \newcommand{\U}{{\mathcal{U}}}
 \newcommand{\V}{{\mathcal{V}}}
 \newcommand{\W}{{\mathcal{W}}}
 \newcommand{\X}{{\mathcal{X}}}
 \newcommand{\Y}{{\mathcal{Y}}}
 \newcommand{\Z}{{\mathcal{Z}}}

\newcommand{\fA}{{\mathfrak{A}}}
\newcommand{\fB}{{\mathfrak{B}}}
\newcommand{\fC}{{\mathfrak{C}}}
\newcommand{\fD}{{\mathfrak{D}}}
\newcommand{\fE}{{\mathfrak{E}}}
\newcommand{\fF}{{\mathfrak{F}}}
\newcommand{\fG}{{\mathfrak{G}}}
\newcommand{\fH}{{\mathfrak{H}}}
\newcommand{\fI}{{\mathfrak{I}}}
\newcommand{\fJ}{{\mathfrak{J}}}
\newcommand{\fK}{{\mathfrak{K}}}
\newcommand{\fL}{{\mathfrak{L}}}
\newcommand{\fM}{{\mathfrak{M}}}
\newcommand{\fN}{{\mathfrak{N}}}
\newcommand{\fO}{{\mathfrak{O}}}
\newcommand{\fP}{{\mathfrak{P}}}
\newcommand{\fQ}{{\mathfrak{Q}}}
\newcommand{\fR}{{\mathfrak{R}}}
\newcommand{\fS}{{\mathfrak{S}}}
\newcommand{\fT}{{\mathfrak{T}}}
\newcommand{\fU}{{\mathfrak{U}}}
\newcommand{\fV}{{\mathfrak{V}}}
\newcommand{\fW}{{\mathfrak{W}}}
\newcommand{\fX}{{\mathfrak{X}}}
\newcommand{\fY}{{\mathfrak{Y}}}
\newcommand{\fZ}{{\mathfrak{Z}}}

\newcommand{\sgn}{\operatorname{sgn}}
\newcommand{\rank}{\operatorname{rank}}
\newcommand{\supp}{\operatorname{supp}}
\newcommand{\dist}{\operatorname{dist}}
\newcommand{\Aut}{\operatorname{Aut}}
\newcommand{\Aff}{\operatorname{Aff}}
\newcommand{\Cknet}{{\mathcal{C}_{{\rm Knet}}}}
\newcommand{\Ckag}{{\mathcal{C}_{{\rm kag}}}}
\newcommand{\GL}{\operatorname{GL}}
\newcommand{\ul}{\underline}

\title[The first-order flexibility of a crystal framework]{The first-order flexibility of a crystal framework}


\author[E. Kastis and S.C. Power]{E. Kastis and S.C. Power}
\address{Dept.\ Math.\ Stats.\\ Lancaster University\\
Lancaster LA1 4YF \\U.K. }
\email{l.kastis@lancaster.ac.uk}
\email{s.power@lancaster.ac.uk}


 \thanks{
Supported by EPSRC grant \quad  EP/P01108X/1 \emph{Infinite bond-node frameworks}}

\date{}

\begin{abstract} 
Four sets of necessary and sufficient conditions are obtained for the first-order rigidity of a periodic bond-node framework $\C$ in $\bR^d$ which is of crystallographic type.
In particular, an extremal rank characterisation is obtained which incorporates a multi-variable matrix-valued transfer function $\Psi_\C(z)$ defined on the product space $\bC^d_* = (\bC\backslash \{0\})^d$. 
In general the first-order flex space 
is shown to be the closed linear span of polynomially weighted geometric velocity fields whose geometric multi-factors in $\bC^d_*$ lie in a finite set.  Paradoxically, first-order rigid  crystal frameworks may possess nontrivial nondifferentiable continuous motions. 
The examples  given are associated with aperiodic displacive phase transitions between periodic states.
\end{abstract}

\maketitle

\section{Introduction}\label{s:intro} Let $\C$ be a periodic bar-joint framework in  $\bR^d$, where $d \geq 2$, which is of  crystallographic type.
The vector space $\F(\C;\bR)$ of real infinitesimal flexes, or first-order flexes,
is the space of $\bR^d$-valued velocity fields  
on the joints of $\C$ which satisfy the first-order flex condition for every bar. 
This space contains the finite-dimensional vector space $\F_{\rm rig}(\C;\bR)$ for rigid body motions 
and, as in the theory of finite bar-joint frameworks (\cite{asi-rot}, \cite{gra-ser-ser}), the crystal framework $\C$ is said to be \emph{infinitesimally rigid}, or \emph{first-order rigid},  if   $\F(\C;\bR) = \F_{\rm rig}(\C;\bR)$. See Owen and Power \cite{owe-pow-crystal}, for example. 
There have been a number of recent theoretical accounts of flexibility and rigidity in infinite periodic structures, such as  \cite{con-she-smi},  \cite{mal-the-1}, \cite{nix-ros},  \cite{sch-tan}. Also  in materials science, over a much longer period,  there have been extensive studies of flexibility, stability and phonon modes, such as \cite{bor}, \cite{dov-exotic}, \cite{gid-et-al}, \cite{kap-et-al}, \cite{weg}.
However these accounts generally assume some form of periodic boundary conditions and so far there has been no characterisation given for first-order rigidity per se.  
In what follows we obtain four sets of
necessary and sufficient conditions by using completely new methods, taken from commutative algebra and algebraic spectral synthesis. 

By simple linearity, the rigidity condition is equivalent to the corresponding equality, $\F(\C;\bC) = \F_{\rm rig}(\C;\bC)$, for complex scalars and so, as usual, we
consider throughout complex velocity fields and complex infinitesimal flexes.
An evident necessary condition  is the triviality of a \emph{geometric flex spectrum}  $\Gamma(\C)$ associated with  $\C$. This is a subset of the product $\bC_*^d=(\bC\backslash \{0\})^d$ which extends the rigid unit mode (RUM) spectrum in the $d$-torus $\bT^d$ which underlies the analysis of low energy phonon modes (mechanical modes) and almost periodic flexes. See \cite{bad-kit-pow}, \cite{bad-kit-pow-2}, \cite{owe-pow-crystal} and  \cite{pow-poly} for example. 
A point $\omega= (\omega_1, \dots , \omega_d)$ in the geometric flex spectrum corresponds to nonzero velocity fields 
which are $d$-periodic modulo the nonzero multiplicative factors given by the components of $\omega$. { We refer to such a velocity field  as a \emph{geometric flex},
or \emph{factor-periodic flex}, with 
\emph{multi-factor} $\omega$. It follows that infinitesimal rigidity implies that the geometric spectrum is trivial in the sense of reducing to the point
$\ul{1}=(1,\dots , 1)$. Additionally, the space of periodic flexes for the periodic structure, taken in the movable lattice sense, must coincide with the $d$-dimensional space of infinitesimal translations. We shall show, in particular, that these two conditions, stated in  condition (iii) of Theorem \ref{t:irigid}, are sufficient as well as necessary. 

Our main approach is to view the geometric flex spectrum in two other ways. Firstly, in difference equation terms, it is the set of solutions of the characteristic equations of a set of linear difference equations, for vector-valued multi-sequences, that arises from a choice of periodic structure for $\C$. These solutions are the points of rank degeneracy of a matrix-valued \emph{transfer function} $\Psi_\C(z)$ on $\bC_*^d$. This in turn can be viewed as the extension of the symbol function $\Phi_\C(z)$, with domain $\bT^d$, associated with rigid unit modes. 
Secondly,  in commutative algebra terms, the geometric flex spectrum is related to the  $\bC[z_1,\dots ,z_d]$-module generated by the rows of the transfer function associated with the periodic structure. 
Our proofs exploit these perspectives together with Noetherian module variants of fundamental arguments in algebraic spectral synthesis which are due to Marcel Lefranc \cite{lef}. In particular we use the Hahn Banach separation theorem for topological vectors spaces of sequences, we appeal to Hilbert's strong Nullstellensatz and Krull's intersection theorem, and we make use of the Lasker-Noether primary decomposition of Noetherian modules.

For a general crystal framework the space of all first-order flexes is invariant under the natural translation operators and is closed with respect to the topology of coordinatewise convergence. It is of significance then to have available a characterisation of general closed shift-invariant subspaces of the space $C(\bZ^d;\bC^r)$ of vector-valued functions on $\bZ^d$. This topic is of independent interest and is the subject of Section \ref{s:lefranc}, which is essentially self-contained with full proofs. In particular, in
  Theorem \ref{t:vectorialsynthesis} we generalise Lefranc's spectral synthesis theorem for $C(\bZ^d)$ 
to this vector-valued setting. 

From these results  we see in Theorem \ref{t:flexthm}  that $\F(\C;\bC)$ 
is the closed linear span of flexes which are {vector-valued polynomially weighted geometric multi-sequences}.
Moreover there is a dense linear span of this type where the associated geometric multi-factors $\omega \in \bC^d_*$ of the velocity fields are finite in number, where this finiteness derive from the Lasker-Noether decomposition 
of a $\bC[z_1,\dots ,z_d]$-module for $\C$. The theorem may thus be viewed as providing an answer, albeit an ambiguous one if the flex space is infinite-dimensional, to the informal question: \emph{What are the fundamental first-order modes of a crystal framework ?} 
Also, it follows that  $\F(\C;\bC)$ is finite-dimensional 
if and only if the geometric spectrum is a finite set.

Paradoxically, a first-order rigid  crystal framework may possess a nontrivial continuous motion, and indeed such a motion is necessarily non-smooth. Our  examples, in Section \ref{s:examples}, follow from elementary geometric arguments associated with aperiodic displacive phase transitions between periodic states.


\section{Preliminaries}\label{s:prelim} A \emph{crystal framework} $\C$ in $\bR^d$ is defined to be a  bar-joint framework $(G, p)$ where $G= (V,E)$ is a  countable simple graph and $p:V \to \bR^d$ is an injective translationally periodic placement of the vertices as joints $p(v)$. 
It is assumed here, moreover, that the periodicity is determined by a basis of $d$ linearly independent vectors and that the corresponding translation classes for the joints and bars are finite in number. The assumption that  
$p:V \to \bR^d$ is injective is not essential although with this relaxation one should assume that each bar $p(v)p(w)$ has positive length $\|p(v)-p(w)\|$.

 The complex infinitesimal flex space $\F(\C;\bC)$ is the vector space of  $\bC^d$-valued functions $u$ on the set of joints satisfying the first-order flex conditions
\[
(u(p(v)) - u(p(w)))\cdot (p(v) - p(w))=0, \quad vw \in E.
\]
Coordinates for this vector space and the space $\V(\C;\bC)$  of all velocity fields may be introduced, first, by making a (possibly different) choice of $d$ linearly independent periodicity vectors for $\C$, which we shall denote as
\[
\ul{a} = \{a_1, \dots , a_d\},
\]
and, second, by choosing finite sets, $F_v$ and $F_e$ respectively, for the corresponding translation classes of the joints  and the bars.  We refer to the basis choice  $\ul{a}$ as a choice of \emph{periodic structure} for $\C$ (following terminology from Delgado-Freidrichs \cite{del}) while the pair $\{F_v, F_e\}$ represents a choice of \emph{motif} for this periodic structure  \cite{owe-pow-crystal}, \cite{pow-poly}. 


\subsection{Transfer functions and $\bC(z)$-modules} Let $\bC[z] = \bC[z_1,\dots ,z_d]$ be the ring of polynomials
in the commuting variable $z_1, \dots , z_d$ over the field $\bC$.  Identify this with the algebra of multi-variable complex polynomials defined on $\bC_*^d$ and write $\bC(z)$ for the containing ring of functions on $\bC_*^d$ generated by  the coordinate functions $z_1,\dots, z_d$  and their inverses $ z_1^{-1},\dots , z_d^{-1}$.
This is the ring of multivariate complex trigonometric polynomials which we shall refer to as the \emph{Laurent polynomial ring}.

Let $n=|F_v|$ and $m=|F_e|$. 
Borrowing terminology from the theory of difference equations we now define the \emph{transfer function} $\Psi_\C(z)$ of $\C$  which is an $m \times dn$ matrix of functions in $\bC(z)$ determined by the pair $\{F_v, F_e\}$. We label the vertices in $V$, and hence the joints $p(v)$ of $\C$, by pairs $(v,k)$ where $p(v,0) = p(v) \in F_v$ and $p(v,k)$, for $k \in \bZ^d$, is the joint
$p(v,0)+ k_1a_1+ \dots + k_da_d$.



\begin{defn}\label{d:transferfunction}
Let $\C$ be a crystal framework in $\bR^d$ with motif 
$\{F_v, F_e\}$ and let $p(e)=p(v,k)-p(w,l)$ be the vector for the bar $p(v,k)p(w,l)$ in $F_e$ associated with the edge $e=(v,k)(w,l)$. 

(i) The \emph{transfer function} $\Psi_\C(z)$ is the $m \times dn$ matrix over the Laurent polynomial ring whose rows are labelled by the edges $e$ for the bars of $F_e$ and whose columns are labelled by the vertices $v$ for the joints of $F_v$ and coordinate indices in $\{1,\dots ,d\}$.
The row for an edge $e= (v,k)(w,l)$ with $v\neq w$ takes the form
\[\kbordermatrix{& & & & v & & & & w & & & \\
e & 0 & \cdots &0 & p(e){z}^{-k} &0& \cdots&0 &- p(e){z}^{-l} &0& \cdots &0 }\]
while if $v=w$ it takes the form
\[\kbordermatrix{& & & & v & & & \\
e & 0 & \cdots &0 & p(e)({z}^{-k} - {z}^{-l}) &0& \cdots&0 }
\]

(ii) The $\bC(z)$-module of $\C$, associated with the motif $\{F_v, F_e\}$, is the submodule 
\[
M(\C) = \bC(z)p_1(z) + \dots + \bC(z)p_m(z)
\] 
of the $\bC(z)$-module $\bC(z)\otimes \bC^{dn}$, where
$p_1(z), \dots , p_m(z)$ are the vector-valued functions given by the rows of the transfer function.

\end{defn}
\medskip

For a given periodic structure one may rechoose the set $F_v$, through an appropriate translation into the positive cone of $\bR^d$, so that the multi-variable vector-valued polynomials $p_i(z)$ are replaced by vector-valued polynomials $z^kp_i(z)$, in $\bC[z]\otimes  \bC^{dn},$ for some fixed $k$. Henceforth we assume that this choice has been made. We may therefore define the $\bC[z]$-module $M(\C)^*$ as the submodule of the left $\bC[z]$-module $\bC[z]\otimes \bC^{dn}$
generated by the vector-valued polynomials $p_1(z), \dots , p_m(z)$. In particular we have
\[
M(\C)^* = M(\C)\cap (\bC[z]\otimes  \bC^{dn}).
\]

Different choices of $F_e$ for the same periodic structure give transfer functions that are equivalent in a natural way. Specifically, the
replacement of a motif edge  by an alternative representative results in the multiplication of
the appropriate row  by a monomial. Also any relabelling of the motif joints and bars corresponds to column and row permutations. It follows that any two transfer functions,
$\Psi_1(z)$ and $\Psi_2(z)$, for a given periodic structure  satisfy the equation
$\Psi_2(z)=D_1(z)A\Psi_1(z)BD_2(z)$,
where $D_1(z)$ and $D_2(z)$ are diagonal monomial matrices and $A, B$ are permutation matrices.

The values $z= \omega$ for which the rank of $\Psi_\C(\omega)$
is less than $dn$ lead to a finite-dimensional space of complex infinitesimal flexes which are periodic up to a multiplicative factor. 
Such flexes are referred to here as  \emph{factor-periodic flexes} since they are characterised by a set of  equations of the form
\[
u_k = \omega^k u_0 = \omega_1^{k_1}\cdots \omega_d^{k_d}u_0,
\]
which relate the (complex) velocity $u_0$ of a joint $p(v)$ in $F_v$
to the velocity $u_k$ of the 
joint $p(v,k)$ for $k\in \bZ^d$. 


\begin{defn}Let $\C$ be a crystal framework in $\bR^d$ with a choice of periodic structure and labelled motif, and associated transfer function $\Psi_\C(z)$.

(i) The \emph{geometric flex spectrum} of $\C$
is the set
$$
\Gamma(\C) = \{\omega \in \bC_*^d: \ker \Psi_\C(\omega^{-1}) \neq \{0\}\}.
$$ 

(ii) The \emph{rigid unit mode spectrum} or {RUM spectrum of $\C$} is the subset $\Omega(\C) = \Gamma(\C)\cap \bT^d$. 
\end{defn}

From our earlier remarks it follows that the sets $\Gamma(\C), \Omega(\C)$ depend only on the choice of periodic structure (up to coordinate relabelling).

The geometric flex spectrum was introduced recently in Badri, Kitson and Power \cite{bad-kit-pow-2} in connection with the existence and nonexistence of bases of localised flexes which generate the entire space of infinitesimal flexes.  We comment more on such bases in Section \ref{s:examples} and Remark \ref{r:bases}.

\subsection{Velocity fields and forms of rigidity}\label{s:fieldsandrigidity} All variants of infinitesimal rigidity depend on a choice of vector space of preferred velocity fields. In this section we define such vector spaces and the resulting forms of periodic and aperiodic infinitesimal rigidity. We first describe a space of \emph{exponential velocity fields} which plays a key role in our main results.

Let $a$ be a vector in $\bC^{dn}$ which is in the nullspace of $\Psi_\C(\omega^{-1})$. Then the function 
\[
u : \bZ^d \to \bC^{dn},\quad  k \to \omega^ka
\]
defines a factor-periodic velocity field which is an infinitesimal flex \cite{bad-kit-pow}, \cite{pow-poly}. {In this  coordinate formalism a complex velocity field for the framework $\C$ is given by a function (or vector-valued multi-sequence) $u$ in $C(\bZ^d;\bC^{nd})$ where $u(k)$ is a combined velocity vector for the $n$ joints which are the translates of the motif joints  by the vector $a(k)=k_1a_1+ \dots + k_da_d$. Explicitly, with $F_v= \{v_1, \dots , v_n\}$, we have
\[
u(k) = (u(p(v_1,k)),\dots , u(p(v_n,k)))
\]  
where $u(p(v_i,k))$ is the velocity vector at the joint $p(v_i,k)=p(v_i)+a(k)$, and where we have introduced notation $(v_i,k)$ for the vertices of the underlying graph $G$.}

We now introduce terminology for factor-periodic velocity fields and related velocity fields.
Let $\omega \in\bC_*^d$ and write $e_\omega \in  C(\bZ^d)$ for the \emph{geometric multi-sequence} given by $e_\omega(k) = \omega^k$, for $k\in \bZ^d$. 
More generally, a \emph{polynomially weighted geometric  multi-sequence}, or $pg$-sequence, is a multi-sequence in  $C(\bZ^d)$ of the form
$e_{\omega, q}: k \to q(k)\omega^k$, where $q(z)$ is a polynomial in $\bC[z]$. 
Define  $\V_{\rm exp}(\C;\bC)$, the space of \emph{exponential velocity fields}, to be the subspace of $\V(\C;\bC)$ formed by  the linear span of the velocity fields
$e_{\omega, q}\otimes a$, for all $\omega$ in $\bC^d_*$, all polynomials $q(z)$ in $\bC[z]$ and all vectors $a$ in 
$\bC^{dn}$. 
This space does not depend on a choice of periodic structure. 

An infinitesimal flex in 
$\V_{\rm exp}(\C;\bC)$ is referred to as an \emph{exponential flex} and these vectors determine a subspace, denoted
$\F_{\rm exp}(\C;\bC)$. That is,
\[
\F_{\rm exp}(\C;\bC) = \V_{\rm exp}(\C;\bC)\cap \F(\C;\bC). 
\]
We say that $\C$ is $\V_{\rm exp}$-rigid, or \emph{exponentially rigid} if $\F_{\rm exp}(\C;\bC)= \F_{\rm rig}(\C;\bC)$.

We next recall various forms of periodic rigidity, each of which is associated with a subspace of $\V_{\rm exp}(\C;\bC)$.

Given a choice of periodic structure for $\C$ define $\V_{\rm per}(\C;\bC)$ to be the associated vector space  of  periodic velocity fields and  write $\F_{\rm per}(\C;\bC)$ for the subspace of  periodic first-order flexes. When there is cause for confusion these flexes are also referred to a \emph{strictly} periodic flexes, with the periodic structure understood. The periodic flexes are the factor-periodic flexes for the multi-factor $\omega = \ul{1}=(1,\dots ,1)$.
The framework $\C$ is said to be  \emph{periodically rigid}, or $\V_{\rm per}$-rigid, if  $\F_{\rm per}(\C;\bC)\subseteq \F_{\rm rig}(\C;\bC)$. The inclusion here is proper since infinitesimal rotations are not periodic infinitesimal flexes. The terms \emph{fixed lattice rigid}, \emph{fixed torus rigid} and \emph{strictly periodically rigid} are also used for this notion of rigidity.

A weaker form of periodic rigidity, known as \emph{flexible lattice periodic rigidity} (and also termed \emph{flexible torus rigidity} or simply \emph{periodic rigidity}) is associated with a larger space of velocity fields $u \in C(\bZ^d;\bC^{nd})$ which have the form
\[
u(k) = u(0) + (Xk,\dots ,Xk), \quad \mbox{where}\quad   X \in M_d(\bC).
\]
These form an $(nd + d^2)$-dimensional space of velocity fields which  are periodic modulo an affine correction in which the $n$ joints in the $k^{th}$-cell each receive an additional velocity $Xk$. We write this space as $\V_{\rm fper}(\C;\bC)$ and note that we have a direct sum
\[
\V_{\rm fper}(\C;\bC):=\V_{\rm per}(\C;\bC)+\V_{\rm axial}(\C;\bC)
\]
where $\V_{\rm axial}(\C;\bC)$ is the space of the \emph{axial velocity fields}, 
$u: k \to (Xk,\dots ,Xk)$. 

\begin{lem}\label{l:inclusions} Let $\C$ be a crystal framework. Then
$ \V_{\rm rig}(\C;\bC)\subseteq\V_{\rm fper}(\C;\bC) \subseteq\V_{\rm exp}(\C;\bC).
$ 
\end{lem}

\begin{proof}
A translational infinitesimal flex $u:k \to \bC^{nd}$, associated with the velocity $b\in \bC^d$, has the form  $e_{\ul{1},q}\otimes (b,\dots ,b)$ with $q(z)$ identically equal to $1$. In particular it is strictly periodic. On the other hand let $u$ be the rotational infinitesimal flex  associated with the orthogonal matrix $B$ in $M_d(\bR)$, let $(p_1,\dots ,p_n)$ be the vector of joints from a motif for the periodic structure $\ul{a}$, and let $A:k \to k_1a_1+\dots + k_da_d$. Then $u(0)=(B(p_1), \dots , B(p_n))=(b_1,\dots ,b_n)$ and
\[
u(k) = (B(p_1+A(k)),\dots , B(p_n+A(k)))
= (b_1+BA(k),\dots , b_n+ BA(k)).
\] 
The right hand expression is linear in $k_1, \dots k_d$ 
and so $u$ may be written in the form 
\[
\sum_{|j|\leq 1} q_j(k)a_j = \sum_{|j|\leq 1} e_{\ul{1},q_j}\otimes a_j
\]
where $a_j\in \bC^{nd}$ and   $q_j(z)$ is the linear polynomial $z^j$ with total degree ${|j|\leq 1}$. From these observations the inclusions follow.
\end{proof}

Let $\F_{\rm fper}(\C;\bC)= \F(\C;\bC)\cap \V_{\rm fper}(\C;\bC)$.
This is the space of flexible lattice periodic flexes for the given periodic structure. It is also referred to as the space of affinely periodic infinitesimal flexes \cite{con-she-smi}, \cite{pow-aff}.

\begin{defn}\label{d:flaxlatticerigid} A crystal framework $\C$ is said to be \emph{flexible lattice periodically rigid}, or $\V_{\rm fper}$-rigid,  if  $\F_{\rm fper}(\C;\bC)= \F_{\rm rig}(\C;\bC)$.
\end{defn}

Let $\ul{1}$ be the point $(1,\dots ,1)$ in $\Gamma(\C)$. A necessary and sufficient condition for periodic rigidity is that the scalar $m \times dn$ rigidity matrix $R_{\rm per}(\C) = \Psi(\ul{1})$ has rank $dn-d$. Borcea and Streinu \cite{bor-str} have obtained an analogous necessary and sufficient condition for flexible lattice periodic rigidity.   Another derivation of this characterisation is in Power \cite{pow-aff}. The rigidity condition is  the maximality of the rank of  a matrix, which we write here as $R_{\rm fper}(\C)$, which is an augmentation of
$R_{\rm per}(\C)$ by $d^2$ columns associated with the entries of the variable matrix $X$, as in the following definition.
The maximal rank condition is then 
\[
\rank R_{\rm fper}(\C) = dn + d(d-1)/2
\]

\begin{defn}
Let $\C$ be a crystal framework in $\bR^d$ with motif 
$(F_v, F_e)$ and let $p(e)=p(v,k)-p(w,l)$ be the vectors associated with the bars in $F_e$ corresponding to edges $e=(v,k)(w,l)$. 
The flexible lattice periodic rigidity matrix $R_{\rm fper}(\C)$ is the $m \times (dn+d^2)$ matrix whose rows, labelled by the edges  $e$ with $v\neq w$, have the form
\[\kbordermatrix{& & & & v & & & & w & & & &  &\\
e & 0 & \cdots &0 & p(e) &0& \cdots& 0 &-p(e) &0& \cdots & (l_1-k_1)p(e) &\cdots &(l_d-k_d)p(e) }\]
while the rows with $v=w$  take the form
\[\kbordermatrix{& & & &  &\\
& 0&\cdots   & 0 &(l_1-k_1)p(e)& \cdots &(l_d-k_d)p(e)}
\]
\end{defn}


\section{The main results}\label{s:mainresults}
In the next section we define a duality between $\bC[z]$-modules in $\bC[z]\otimes \bC^{dn}$ and shift-invariant subspaces of the space $C(\bZ^d;\bC^{dn})$ of velocity fields. It is this duality that underlies the following definition.

\begin{defn} The \emph{$\bC[z]$-rigidity module} $M_{\rm rig}(d,n)$ associated with a periodic structure for $\C$, and which is also denoted $M_{\rm rig}(\C)$ when the periodic structure is understood, is the annihilator of the space $\F_{\rm rig}(\C;\bC)$ in $\bC[z]\otimes \bC^{dn}$.
\end{defn}

We say that a transfer function $\Psi_\C(z)$ is  \emph{rank extremal} if $\rank \Psi_\C(z) = dn$ for all $z\in \bC_*^{d}$ and the rank of $\Psi_\C(\ul{1})$ is $dn-d$. Also we say that $R_{\rm fper}(\C)$ is \emph{rank extremal} if its rank is $dn + d(d-1)/2$.

\begin{thm}\label{t:irigid}
The following statements are equivalent for a crystal framework $\C$ in  $\bR^d$. 
\medskip 

(i) $\C$ is first-order rigid.
\medskip

(ii)  $\C$ is exponentially rigid.
\medskip

(iii) For a given periodic structure $\C$  there are no nontrivial factor-periodic flexes or flexible lattice periodic flexes.
\medskip

(iv)  For a given periodic structure and motif the transfer function $\Psi_\C(z)$  and the matrix $R_{\rm fper}(\C)$
are rank extremal.

\medskip

(v) For a given periodic structure the $\bC(z)$-module $M(\C)$ agrees with the rigidity module $M_{\rm rig}(\C)$.


\end{thm}


\begin{defn}
Let $\C$ be a crystal framework  in  $\bR^d$ with a periodic structure with $n$ translation classes of joints. A \emph{vectorial $pg$-sequence} for $\C$, for this periodic structure, with \emph{geometric index} $\omega \in \bC^d_*$, is a velocity field $u_{\omega, {h}}:\bZ^d \to \bC^{nd}$ of the form
\[
u_{\omega, {h}}: k \to \omega^kh(k) 
\]
where $h(z)$ is a vector-valued polynomial in $\bC[z]\otimes \bC^{dn}$.
\end{defn}

The term \emph{root sequence}  in the following theorem is defined in Section \ref{ss:noetherian} while the term \emph{closed} refers to the topology for coordinate-wise convergence or, more precisely, the topology of pointwise convergence in the space of velocity fields.

\begin{thm}\label{t:flexthm} Let $\C$ be a crystal framework in  $\bR^d$ with a given periodic structure and associated $\bC[z]$-module 
$M(\C)^*$. Then $\F(\C;\bC)$ is the closed linear span
of $pg$-sequences $u_{\omega, {h}}$ in $\F(\C;\bC)$. Moreover, if $\omega(1),\dots ,\omega(s)$ is a root sequence for the Lasker-Noether decomposition of $M(\C)^*$ then  $\F(\C;\bC)$ is the closed linear span of the $pg$-sequences $u_{\omega, {h}}$ in $\F(\C;\bC)$ with geometric indices $\omega(1),\dots ,\omega(s)$.
\end{thm}


It is possible, although unusual, for a crystal framework to have a finite-dimensional first-order flex space which is strictly larger than the finite-dimensional space of rigid motion flexes, and we give some examples below. The finiteness of the geometric spectrum is a simple necessary condition for this phenomenon and we shall show, from the primary decomposition structure of $M(\C)^*$, that it is also a sufficient condition.

\begin{thm}\label{t:findimflex}  Let $\C$ be a crystal framework in  $\bR^d$ with a given periodic structure and associated geometric flex spectrum $\Gamma(\C)$. Then the following statements are equivalent.

(i) $\F(\C;\bC)$ is finite-dimensional.

(ii) $\Gamma(\C)$ is a finite set.

\end{thm}

\begin{rem}\label{r:unbounded}
In view of the unbounded nature of a $pg$-flex with nonunimodular geometric multi-factor it might appear that these results 
have little relevance to materials science. This is definitely not the case however since \emph{surface modes}, associated with a hyperplane boundary wall for example, arise as bounded restrictions of unbounded flexes of the bulk crystal. See for example  Lubensky et al \cite{lub-et-al}, Power \cite{pow-seville},  Rocklin et al \cite{roc-et-al} and Sun et al \cite{sun-et-al}. Thus
the geometric spectrum in effect identifies free surfaces where one may find surface modes with geometric decay into the bulk, and the  $\bC[z]$-module $M(\C)^*$ and its annihilator provides further information.


The commutative algebra viewpoint also usefully extends the conceptual analysis of crystal frameworks which, for example, may now be described as  \emph{primary} or \emph{properly decomposable} according to whether these properties hold for the $\bC[z]$-module $M(\C)^*$ associated with a primitive periodic structure.
\end{rem}

 In the rest of this section we recall the Lasker-Noether theorem, we give some simple examples to illustrate the main results and some steps of the proofs, and we discuss primary ideals in formal power series rings.

\subsection{The primary decomposition of modules for Noetherian rings}\label{ss:noetherian}
The Lasker-Noether theorem states that every submodule of a finitely generated module over a Noetherian ring is a finite intersection of {primary} submodules. 


\begin{defn}Let $R$ be a Noetherian ring, let $L$ be a submodule of an $R$-module $N$, and for $p\in R$, let $\lambda_p:N/L \to N/L$ be multiplication by $p$. Then $L$ is a \emph{primary submodule} of $N$ if $L$ is proper and for every $p$ the map $\lambda_p$ is either injective or nilpotent. 
If $P=\{p\in R: \lambda_p \mbox{   is nilpotent}\}$ then $P$ is a prime ideal and  $L$ is said to be a \emph{$P$-primary submodule} of $N$.
\end{defn}





\begin{defn}\label{d:roots}
Let $M = Q_1\cap \dots \cap Q_s$ be a primary decomposition of the $\bC[z]$-module $M$ where $Q_i$ is $P_i$-primary for distinct primes $P_i, 1\leq i\leq s$.
 A \emph{root sequence} for $M$ is a set $\omega(1), \dots ,\omega(s)$ of points in $\bC^d$ where for each $1 \leq i \leq s$ the point $\omega(i) $ is a \emph{root} of $P_i$ in the sense that $p(\omega(i))=0$ for all $p(z)$ in $P_i$.
\end{defn}

{ For more details and discussion see  Ash  \cite{ash}, as well as Atiyah and MacDonald \cite{ati-mac}, Krull \cite{kru} and Rotman \cite{Rot10}. 
In particular (Chapter 1 of \cite{ash}) a strong form of the Lasker-Noether theorem asserts that every finitely generated submodule $M$ of a Noetherian ring has a decomposition as given in Definition \ref{d:roots}, and this is called a \emph{primary decomposition}. {Moreover any  such decomposition leads to a \emph{reduced} primary decomposition with distinct primes ideals $P_i$, and this set of prime ideals  is uniquely determined by $M$.} 



\subsection{Examples and remarks}\label{s:examples} Consider first the simplest possible connected crystal framework in two dimensions, namely the 2D grid framework $\C_{\bZ^2}$, whose joints lie on the integer lattice. We show that there is a set of vectorial $pg$-sequences with dense span in the flex space. 

For a suitable choice of single joint motif, the transfer function $\Psi(z)$ has 2 row vector functions, $p_1(z_1,z_2) = (1-z_1,0), p_2(z_1,z_2) = (0, 1-z_2)$. The corresponding $\bC[z]$-module in $\bC[z]\otimes \bC^2$ is
\[
M^*= \bC[z]p_1(z)+\bC[z]p_2(z) =  (\bC[z](1-z_1), \bC[z](1-z_2)).
\]
Consider $Q_1^*=(\bC[z](1-z_1), \bC[z])$ and $Q_2^*= (\bC[z], \bC[z](1-z_2))$. Then $ M^*= Q_1^*\cap Q_2^* $. Moreover $M^*$ is a submodule of $N=\bC[z]\otimes \bC^2$ and $N/Q_1^*$ is module-isomorphic to $\bC[z]/(1-z_1)\bC[z]$. Thus, for $p(z) \in \bC[z]$ the map $\lambda_p$ is injective if  $(1-z_1)$ is not a factor of $p(z)$ and is zero otherwise. Thus $Q_1^*$, and similarly $Q_2^*$, are primary submodules of $N$. {Also the ideals $P_1=(1-z_1)\bC[z]$  and $P_2=(1-z_2)\bC[z]$) are prime ideals in  $\bC[z]$, and in fact they are associated prime ideals in  $\bC[z]$ for $Q_1^*$ and $Q_2^*$, respectively,  and $Q_i$ is $P_i$-primary.}

We now see that a root sequence $\{\omega(1), \dots ,\omega(s)\}$, for  
$M^*$, can be any pair $\{(1,\xi_2), (\xi_1,1)\}$ with $\xi_1, \xi_2$ in $\bC_*$. 
Let us take $\omega(1)= (1,1), \omega(2)= (1,1)$. Theorem \ref{t:flexthm} predicts that there is a set of infinitesimal flexes of the form
\[
k\to (h_1(k), h_2(k)),\quad  h_1(z), h_2(z)\in \bC[z],
\]
whose closed linear span is $\F(\C_{\bZ^2};\bR)$.
To see, independently, that this is true consider first the polynomials
$h_1(z)$ of the form $h_1(z_1, z_2) = h(z_2)$, with $h(z)$ a single variable polynomial. The velocity field
\[
u: k \to (h_1(k),0) 
\]
is a velocity field which gives a constant horizontal velocity to the joints on each horizontal line. These are infinitesimal flexes.
Moreover it is straightforward to show by direct arguments that the closed span of these flexes give the space of \emph{all} infinitesimal flexes with this horizontal constancy property, including, in particular, the localised translational flexes which are supported on a single horizontal line of joints. We remark that these localised translational flexes are evidently not in the (unclosed) linear span of vectorial $pg$-sequences.

Exchanging the roles of the variables it follows similarly that there are vectorial $pg$-flexes whose closed linear span contains the vertically localised flexes. The closed span of the vertically localised flexes and the horizontally localised flexes is the space of all flexes, since one can show that every infinitesimal flex is an \emph{infinite} linear sum of the line-localised flexes. Thus the conclusion of the theorem is confirmed for $\C_{\bZ^2}$.

Another favoured crystal framework example in $\bR^2$ is the \emph{kagome framework}, $\C_{\rm kag}$, which is associated with the regular hexagonal tiling of the plane. There are $3$ joints and $6$ bars in a primitive motif and the $\bC[z]$-module $M(\C_{\rm kag})^*$ in $\bC[z]\otimes \bC^6$ has decomposition length $s=3$. A direct verification of the conclusion of Theorem \ref{t:flexthm} may be obtained as above by exploiting the fact that the  
line-localised  flexes form a generalised basis for the flex space \cite{bad-kit-pow-2}.

A simple 2-dimensional crystal framework which illustrates Theorem \ref{t:findimflex} may be obtained from $\C_{\bZ^2}$ by adding the diagonal bars $(n,m)(n+1,m+1)$, for $n+m$  even. This is Example 3 from the gallery of examples in Badri, Kitson and Power \cite{bad-kit-pow}. An explicit primitive motif consists of 2 joints and 5 bars and the RUM spectrum and the geometric spectrum are equal to the set $\{(1,1),(-1,1)\}$. One can verify directly that the first-order flex space is $4$-dimensional and is spanned by a basis for the rigid motion flexes together with a single geometric flex, with $\omega = (-1,1)$, that restricts to alternating rotational flexes of each diagonalised square subframework. 

For a simple $3$-dimensional illustration, with $\Gamma(\C)=\{\ul{1}\}$, one may take an infinitesimally rigid crystal framework and attach a parallel copy (with the same period vectors) by means of parallel bars between corresponding joints. In this case the first-order flex space has dimension $8$. More elaborate (connected) examples of the same flavour may be obtained from (disconnected) entangled frameworks \cite{car-cia-pro2003}
by the periodic addition of connecting bars.

\begin{rem}\label{r:APflexes}
While the pure geometric flexes alone need not have dense span in the flex space they may nevertheless be sufficient for restricted classes of first-order flexes with respect to other closure topologies. This has been shown to be the case for the space of uniformly almost periodic flexes  \cite{bad-kit-pow}. It would be of interest to develop further such analytic spectral synthesis and to find spectral integral representations for other classes of flex spaces.
\end{rem}

\begin{rem}\label{r:bases} The existence of generalised bases of {localised geometric flexes} for a crystal framework is considered in Badri, Kitson and Power \cite{bad-kit-pow-2}. It seems, as in the case of the kagome framework for example, that such bases give the best way of understanding the first-order flex space and rigid unit modes in that every such flex is an infinite linear combination of basis elements. However such crystal flex bases need not exist and the considerations in \cite{bad-kit-pow-2} suggest that this is typical unless the geometric spectrum has sufficient linear structure. 


\end{rem}

\begin{rem}\label{r:ulatrarigid}
One can also consider forms of rigidity, which one might call persistent rigidity, with respect to \emph{all} periodic structures, both in the strict (fixed lattice) case and the flexible lattice case. The latter form is known as \emph{ultrarigidity} (see Malestein and Theran \cite{mal-the-2}) while the former form we refer to as \emph{persistent periodic rigidity}. Each may be defined in terms of the vector space of velocity fields which is the union over all periodic structures of the appropriate spaces of periodic velocity fields. These rigidity notions are weaker than strict periodic infinitesimal rigidity but stronger than first-order rigidity.

For a periodic structure for $\C$ define the \emph{rational RUM spectrum} $\Omega_{\rm rat}(\C)$ 
to be the intersection of $\Omega(\C)$ with the 
points in $\bT^d$
whose arguments are rational multiples of $2\pi$.
Then it can be shown that a crystal framework $\C$ is persistently periodically rigid if and only if the matrix values of the transfer function  on the subset $\Omega_{\rm rat}(\C)$ have extremal rank.
The analogous characterisation for ultrarigidity, together with detailed algorithmic considerations, is given in \cite{mal-the-2}.
\end{rem}

\subsection{Primary ideals in $\bC[[z]]$}
In this section we show that a primary $\bC[z]$-module in $\bC[z]\otimes \bC^r$ with root $0$ may be recovered from the  $\bC[[z]]$-module that it generates in $\bC[[z]]\otimes \bC^r$, where $\bC[[z]]$ is the ring of formal power series in $z_1, \dots ,z_n$. This connection plays a key role in our main proof, as we discuss in Section \ref{ss:duality}. Since we have not found a satisfactory reference we give the details of this connection  in Proposition \ref{p:bigideal} and its proof.

Write $\bC[z]_{(z)}$ for the ring of rational functions in $z$ that are continuous on some neighbourhood of $0$. (The notation 
reflects the fact that if $(z)$ is the ideal in $\bC[z]$ generated by $z_1,...,z_d$ then the set $S=\bC[z]\backslash (z)$ is multiplicative and $\bC[z]_{(z)}$ is the localization $S^{-1}\bC[z]$.) Since $(z)$ is maximal and therefore prime, $\bC[z]_{(z)}$ is a Noetherian local ring with unique maximal ideal $m_{(z)}=(z)\bC[z]_{(z)}$. 
We also write $\bC[[z]]$ for the formal power series ring which is also a Noetherian local ring, with unique maximal ideal $m_{[z]}=(z)\bC[[z]]$.  Thus we have the natural ring inclusions
\[
\bC[z] \subset \bC[z]_{(z)} \subset  \bC[[z]].
\]
That these rings are Noetherian is discussed in Atiyah and MacDonald \cite{ati-mac}, for example.

Let $Q$ be a finitely generated submodule of $\bC[z]\otimes \bC^r$, let
$R[Q]:=\bC[z]_{(z)}\cdot Q$ be the corresponding
$\bC[z]_{(z)}$-module in $\bC[z]_{(z)}\otimes \bC^r$, and let $S[Q]= \bC[[z]]\cdot Q$ be the corresponding $\bC[[z]]$-module in $\bC[[z]]\otimes \bC^r$.

\begin{prop}\label{p:bigideal}
Let $Q$ be a primary submodule in $\bC[z]\otimes \bC^r$ with associated root $0$. Then $Q= S[Q]\cap (\bC[z]\otimes \bC^r)$.
\end{prop}

For the proof we use a preliminary lemma which depends on the following  Krull intersection theorem 
\cite{ash}, \cite{ati-mac}.

\begin{thm}\label{t:krullintersection}
Let $R$ be a Noetherian local ring with maximal ideal $m$ and let $N$ be a finitely generated $R$-module. Then $\bigcap_{n=1}^\infty m^nN=\{0\}$.
\end{thm}

\begin{lem} Let $Q$ be a $\bC[z]$-module in $\bC[z]_{(z)}\otimes \bC^r$. Then 
$R[Q]=S[Q]\cap (\bC[z]_{(z)}\otimes \bC^r)$.
\end{lem}

\begin{proof} The inclusion of $R[Q]$ in the intersection is elementary. On the other hand the intersection is equal to the set 
\[
\left\{ P= \sum_{i=1}^N g_if_i \in \bC[z]_{(z)}\otimes \bC^r: g_i\in \bC[[z]], f_i\in Q \right\}.
\]
Write $g_i = g_{i,0} +r_i$ where $g_{i,0}$ is the partial sum of the series for $g_i$ for terms of total degree less than $M$. Then
the element $P_0= \sum_i g_{i,0}f_i$ belongs to $R[Q]$. Also
the element $P_r= \sum_i r_if_i = P - \sum_i g_{i,0}f_i$ belongs to $\bC[z]_{(z)}\otimes \bC^r$. Observe that $P_r$ also belongs to $m_{[[z]]}^M\otimes \bC^r$ and so it belongs to  $m_{(z)}^M\otimes \bC^r$. Thus $P$ lies in the intersection 
\begin{equation}\label{e:intersection}
\bigcap_{M=0}^\infty (R[Q]+ m_{(z)}^M\otimes \bC^r).
\end{equation}
By the Krull intersection theorem 
\[
\bigcap_{M=0}^\infty m_{(z)}^M((\bC[z]_{(z)}\otimes \bC^r)/R[Q]) =\{0\}
\]
and so the intersection of (\ref{e:intersection})
is equal to $R[Q]$, and the lemma follows. 
\end{proof}

\begin{proof}[Proof of Lemma \ref{p:bigideal}] By the previous lemma it suffices to show $Q$ is equal to $R[Q]\cap (\bC[z]\otimes \bC^r)$, which is the set 
\[
\left\{h=\sum g_if_i \in 
\bC[z]\otimes \bC^r : g_i \in \bC[z]_{(z)}, f_i\in Q\right\}.
\]
Let $h$ belong to this set. Then
$h$ is equal to the finite sum $\sum_i\frac{p_i}{q_i}f_i = {\sum a_if_i}/{\prod q_i},$ where $p_i, q_i \in \bC[z]$ for all $i$. Thus
$\sum_i a_if_i = ({\prod q_i})h \in Q$. 

On the other hand, since $Q$ is a primary $\bC[z]$-module, the map 
\[
\lambda_{\prod q_i}:
(\bC[z]\otimes \bC^r)/Q \to (\bC[z]\otimes \bC^r)/Q
\]
is either nilpotent or injective. Since ${\prod q_i}$ does not vanish at the origin the map is not nilpotent and so it follows that $h\in Q$. 
\end{proof}

\section{Shift-invariant  subspaces of $C(\bZ^d; \bC^r)$}\label{s:lefranc}
Let $r\geq 1$ and let $C(\bZ^d; \bC^r)$ be the topological vector space of vector-valued functions  $u: \bZ^d \to \bC^r$ with the topology of coordinatewise convergence. Let $e_1, \dots ,e_d$ be the generators of $\bZ^d$ and let  $W_i, 1\leq i \leq d,$ be the forward shift operators, so that $(W_iu)(k) = u(k-e_i)$, for all $k$ and each $i$.
A subspace $A$ of $C(\bZ^d;\bC^r)$ is said to be an \emph{invariant subspace} if it is invariant for the shift operators and their inverses, or equivalently if $W_iA = A$ for each $i$.
In this section we obtain a spectral synthesis property for  closed shift-invariant subspaces of $C(\bZ^d; \bC^r)$.

\subsection{$\bC(z)$-modules and their reflexivity}\label{ss:duality} There is a bilinear pairing
$\langle p, u\rangle: \bC(z) \times C(\bZ^d) \to \bC$ such that, for $p(z) = \sum_k a_kz^k$ in  $\bC(z)$ and $u=(u_k)_{k\in \bZ^d}$ in $C(\bZ^d)$,
$
\langle p, u\rangle = \sum_k a_ku_k.
$
Similarly, considering $C(\bZ^d; \bC^r)$ as the space $C(\bZ^d)\otimes \bC^r$, for $p=(p_i)\in   \bC(z)\otimes \bC^r $ and $u=(u_i)\in C(\bZ^d)\otimes \bC^r$ we have the corresponding pairing
$\langle p, u\rangle: \bC(z)\otimes \bC^r \times C(\bZ^d)\otimes \bC^r \to \bC$, where
\[
\langle p, u \rangle= \langle (p_i), (u_i) \rangle = \sum_{i=1}^r \langle p_i, u_i\rangle.
\]

It is elementary to show that with this pairing
the vector space dual of
$C(\bZ^d)\otimes \bC^r$ can be identified with 
$\bC(z)\otimes \bC^r$. Also, with the same pairing
the dual space of the vector space $\bC(z)\otimes \bC^r$ is identified with $C(\bZ^d)\otimes \bC^r$.
Thus both spaces are reflexive, that is, equal to their double dual, in the category of vector spaces.
These dual space identifications also hold in the category of linear topological spaces when each is endowed with the topology of coordinatewise convergence, simply because all linear functionals are automatically continuous with these topologies. 


For a subspace $A$ of $C(\bZ^d)\otimes \bC^r$ we write ${B} = A^\perp$ for the annihilator in  $\bC(z)\otimes \bC^r$ with respect to the pairing. Thus 
\[
{B} = \{p \in \bC(z)\otimes \bC^r: \langle p, u\rangle =0, \mbox{ for all } u \in A\}.
\]
Similarly for a subspace ${B}$ of $\bC(z)\otimes \bC^r$ we write ${B}^\perp$ for the annihilator in $C(\bZ^d)\otimes \bC^r$ with respect to the same pairing.

\begin{lem}\label{l:doubleperp}
Let $A$ be a closed subspace of $C(\bZ^d)\otimes \bC^r$ and let $M$ be a closed subspace of $\bC(z)\otimes \bC^r$. Then $A = (A^\perp)^\perp$ and $M = (M^\perp)^\perp$.
\end{lem}

\begin{proof}This follows from the dual space identifications and from the Hahn-Banach theorem for topological vector spaces
 (\cite{Con90}, IV. 3.15).
\end{proof}

The following lemma provides a route for the analysis of shift-invariant subspaces $A$ in terms of the structure of their uniquely associated $\bC(z)$-modules $B=A^\perp$. Note that it follows from the Noetherian property that $\bC(z)$-modules in $\bC(\bZ^d)\otimes \bC^r$ are necessarily closed.

\begin{lem}\label{l:iffmodule}
A closed subspace $A$  in $C(\bZ^d)\otimes \bC^r$ is an invariant subspace if and only if $A^\perp$ is an $\bC(z)$-module of the {module} $\bC(z)\otimes \bC^r$. 
\end{lem}

\begin{proof}
For all $a \in A$, $b \in B=A^\perp$ and $1\leq i\leq d$ we have $\langle W_ia,b\rangle  = \langle a,z_ib\rangle$ and the lemma follows.
\end{proof}

\subsection{Primary decompositions of Noetherian modules} The Lasker-Noether theorem for a nonzero finitely generated module $M$ over a Noetherian ring $R$ ensures that $M$ is an intersection of a finite sequence of primary modules, $Q_1, \dots ,Q_s$, { where $Q_i$ is $P_i$-primary for distinct prime ideals $P_, \dots ,P_s$.}
In particular this decomposition applies to  the crystal framework  module $M(\C)^*$ over the Noetherian polynomial ring $\bC[z]$. The next lemma shows that it is also applicable to the $\bC(z)$-module $M(\C)$.

\begin{lem} The Laurent polynomial ring $\bC(z)$ is a Noetherian ring.
\end{lem}

\begin{proof}  The argument is elementary. 
(Alternatively, if $S$ is the multiplicative subset $\{z^k:k\in \bZ^d_+\}$ then the ring  $\mathbb{C}(z)$ is isomorphic to the localization $S^{-1} \mathbb{C}[z]$, and so is Noetherian by \cite{Rot10}, Corollary 10.20.)
\end{proof}
\medskip

For the rest of this section we let ${B}$ be a proper $\bC(z)$-module in $\bC(z)\otimes \bC^r$ with primary decomposition 
\[
{B} = Q_1\cap \dots \cap Q_s
\]
{as above, where the $\bC(z)$-modules $Q_i$ are $P_i$-primary.}

\begin{lem} Fix $i$, with $1 \leq i \leq s$. Then there exists a point $\omega(i)\in \bC_*^d$ such that if $p(z)$ is a polynomial in $P_i^*=P_i \cap \bC[z]$ then $p(\omega(i))=0$.
\end{lem}

\begin{proof}
To see this note that  the complex variety $V(P_i^*)$ is nonempty by Hilbert's  Nullstellensatz \cite{ati-mac}, since $P_i$ and hence $P_i^*$ is a proper ideal. Moreover, there is a point $\omega(i)$ in this variety which is in $\bC_*^d$. Indeed, if this were not the case then the monomial $z_1\cdots z_d$ would be {zero} on the variety of $P_i$. It then follows from the strong Nullstellensatz (\cite{Rot10}, Theorem 5.99) that for some index $\rho$ the power $(z_1z_2\dots z_d)^\rho $ is in $P_i^*$. This implies $P_i = \bC(z)$ which is a contradiction.
\end{proof}

Write ${B^*}$ for the $\bC[z]$-module  ${B}\cap (\bC[z]\otimes \bC^r)$ and note that $B$ is recoverable from $B^*$ as the set of elements $z^kp(z)$ with $p(z)$ in $B^*$ and $k\in \bZ^d$.
It follows from this that we have the decomposition
\[
{B}^* = Q_1^*\cap \dots \cap Q_s^*
\] 
where the implied modules $Q_i^*$ (the intersections $Q_i \cap (\bC[z]\otimes \bC^r$)) are {$P_i^*$-primary} $\bC[z]$-modules with distinct prime ideals $P_i^*$. Moreover  each prime ideal $P_i^*$ has a root $\omega(i)$ in $\bC^d_*$ (rather than $\bC^d$).
\medskip

\subsection{Modules and dual spaces for power series rings} 
 For each $i = 1, \dots ,s$ and associated root $\omega(i) \in \bC^d_*$, as above, let
${Q_i}^{*b}$ be the ``big" $\bC_{\omega(i)}[[z]]$-module 
generated by the module ${Q_i}^*$, where $\bC_{\omega(i)}[[z]]$ is the ring of {formal power series} in the variables $z_1-\omega(i)_1, z_2-\omega(i)_2, \dots ,$ $ z_d-\omega(i)_d$.
Since $Q_i^*$ is a primary module for the polynomial ring $\bC[z]$ with root $\omega(i)$ it follows from  Proposition \ref{p:bigideal} that $Q_i^* = {Q_i}^{*b}\cap (\bC[z]\otimes \bC^r)$.

Thus  $B^*$ is the set of polynomials $p(z)$ in $\bC[z]\otimes \bC^r$ which lie in the big module ${Q_i}^{*b}$ for each $i$, and so
\begin{equation}\label{bigmoduledecomp}
B^* = (Q_1^{*b}\cap (\bC[z]\otimes \bC^r))\cap \dots \cap (Q_s^{*b} \cap (\bC[z]\otimes \bC^r)).
\end{equation}

The reason for the introduction of this decomposition is that the rings
$\bC_{\omega(i)}[[z]]$ and their finitely generated modules in $\bC_{\omega(i)}[[z]]\otimes \bC^r$ have dual spaces consisting of  \emph{finitely} supported functionals. This follows in the same way as the duality between $\bC(z)$ and $C(\bZ^d)$. At the same time these finitely supported functionals may be represented in different ways, as we see in Proposition \ref{p:diffop}.

We first recall Lefranc's differential operator formalism for scalar-valued trigonometric polynomials, as expressed in the next lemma.

Let $s_i\in \bN$ and let $z_i^{[s_i]}=(z_i+1)(z_i+2)\dots (z_i+s_i)$. A polynomial $q(z)\in \bC[z]$ may be written uniquely as
\[
q(z) = \sum \beta_jz^{[j]}
\]
where $[j] =([j_1],\dots , [j_d])$ and $(\beta_j)$ is a finitely nonzero multi-sequence  with support in $\bZ_+^d$.

\begin{lem}
Let $p(z) = \sum a_kz^k\in \bC[z]$ and let $e_{\omega,q}$ be a $pg$-sequence in $C(\bZ^d)$. Then 
\begin{equation}\label{pg_omega_equn}
\langle p(z), e_{\omega,q} \rangle
= \sum_k a_kq(k)\omega_1^{k_1}\cdots \omega_d^{k_d} 
= \left[\sum \beta_j\partial_j (p(z)z^j)\right]_{z=\omega}
\end{equation}
where $\partial_j$ is the partial derivative for the multi-index $j \in \bZ_+^d$. 
\end{lem}

\begin{proof}Note first that for $p(z) = z^l$, a monomial in $\bC[z]$, we have
\[
\partial_j(p(z)z^j)=\partial_j(z^lz^j) = [\prod_{i=1}^d(l_i+j_i)(l_i+j_i-1)\cdots (l_i+1)]z^l = l^{[j]}z^l. 
\]
Thus, for $q(z)= z^{[j]}$ we have
\[
[\partial_j(p(z)z^j)]_{z=\omega}= q(l)\omega^l = \langle z^{l}, (q(k)\omega^k) \rangle  = \langle p(z), e_{\omega,q} \rangle.
\]
(The pairing here is for $\bC(z)$ and its dual space although we are restricting consideration to polynomials $p(z)$.)
Since the partial differential operators are linear on $\bC[z]$ it follows that the right hand side of the desired equality is linear in $p(z)$. It then follows, by linearity, that the equality holds also for general polynomials $q(z)$.
\end{proof}

For $\omega \in \bC^r_*$ and $(\beta_j)$ a finitely nonzero sequence let us write
 $L_{\omega,\beta}$ for the \emph{differential operator functional} on 
the vector space $\bC_\omega[[z]]$ of formal power series in $z_1-\omega_1, \dots , z_d-\omega_d$ which is
given by
\[
L_{\omega,\beta}: s(z) \to    \left[\sum \beta_j\partial_j (s(z)z^j)\right]_{z=\omega}.
\]

\begin{prop}\label{p:diffop}
The vector space dual of the power series ring $\bC_\omega[[z]]$ is the space of differential operator functionals $L_{\omega,\beta}$.
\end{prop}

\begin{proof}
The dual space of  $\bC_\omega[[z]]$ is the space of finite linear combinations of the natural coefficient functionals. Thus it will be enough to show that for each $j$ the $j^{th}$-coefficient evaluation functional $F_j,$ for $j\in \bZ_+^d$, is given by a differential operator functional $L_{\omega,\beta}$.
Here $F_j$ is defined by linearity and the requirement, in multinomial notation, is that $
F_j((z-\omega)^k) = \delta_{j,k}$ for $  k\in \bZ_+^d. $
Order $\bZ_+^d$ and the corresponding monomials lexicographically. Evidently for $j = (0,\dots ,0)$ the first functional $F_j$ is a differential operator functional. We argue by induction on the lexicographic order. Fix $l\in \bZ_+^d$ and let $\beta$ be the sequence $(\delta_{l,k})_k$. Then 
\[
\left[\sum \beta_j\partial_j (s(z)z^j)\right]_{z=\omega}
= \partial_l (s(z)z^l)_{z=\omega}
=(\partial_l s)(\omega)\omega^l +F(s(z))
\]
where $F$ is a linear functional which is in the linear span of the functionals $F_j$ where $j<l$.
Thus
\[
L_{\omega,\beta}(s(z)) = cF_l(s(z)) + F(s(z))
\] 
where $c=\omega^l$ is nonzero and it follows from the induction hypothesis that $F_l$ has the desired form.
\end{proof}

Returning to vector-valued polynomials note that the vector space dual $(\bC_{\omega}[[z]]\otimes \bC^r)'$
is naturally identifiable with $(\bC_{\omega}[[z]]')\otimes \bC^r$ where
$\bC_{\omega}[[z]]'$ is the dual space of $\bC_{\omega}[[z]]$. Thus we can identify
$(\bC_{\omega}[[z]]\otimes \bC^r)'$ with the space of $r$-tuples
\[
L_{\omega,\ul{\beta}}= (L_{\omega,\beta^1},\dots , L_{\omega,\beta^r})
\]
associated with the set of finite multi-sequences
${\ul{\beta}}= ({\beta^1},\dots , {{\beta^r}})$ where each $\beta^i= (\beta^i_k)$ is a finitely nonzero multi-sequence. 
The vector version of equation \eqref{pg_omega_equn} takes the form
\begin{equation}\label{vector_pg_equn}
\langle p(z), u_{\omega,\ul{q}} \rangle  =  L_{\omega,\ul{\beta}}(p), \quad p(z) \in \bC[z]\otimes \bC^r,
\end{equation}
and in view of Proposition \ref{p:diffop} we can extend this pairing to a pairing
\[
\langle \cdot , \cdot \rangle_\omega : (\bC_{\omega}[[z]]\otimes \bC^r) \times \{u_{\omega,\ul{q}}:\ul{q}(z)\in \bC[z]\otimes \bC^r\} \to \bC
\] 
by {defining}
\begin{equation}\label{vector_pg_equnExtended}
\langle s(z), u_{\omega,\ul{q}} \rangle_\omega  :=  L_{\omega,\ul{\beta}}(s(z))
\end{equation}
where $\ul{q}=(q_1(z),\dots ,q_r(z))$ is the vector of polynomial associated with $\ul{\beta}$ and  $s(z)\in \bC_\omega[[z]]\otimes \bC^r$. 
In this way we describe the dual of the power series space $\bC_{\omega}[[z]]\otimes \bC^r$ in terms which extend the pairing of the submodule $\bC[z]\otimes \bC^r$  with vectorial $pg$-sequences (rather than in terms of sequences with finite support).

The next lemma follows readily as a corollary of Proposition \ref{p:diffop} and the previous observations and is a module variant of a key lemma in Lefranc's argument \cite{lef} for ideals. The term ``orthogonal" is in reference to the extended bilinear pairing above in the case $\omega = \omega(i)$.


\begin{lem}\label{l:mainlemma}Let $\omega(i)$ be a root in $\bC^d_*$ for $Q_i^*$, as above.
Then a polynomial $p(z)$ in $ \bC[z]\otimes \bC^r$ belongs to $Q_i^*$
if and only if it is orthogonal to each vectorial $pg$-sequence $u_{\omega(i),h}$
which is orthogonal to  $Q_i^{*b}$.
\end{lem}

\begin{proof}
 Let $p(z)$ be a polynomial in  $\bC_{\omega(i)}[[z]] \otimes \bC^r$ that is orthogonal to all vectorial $pg$-sequences that are orthogonal to $Q_i^{*b}$. Suppose that $p(z)$ is not in 
$Q_i^*=Q_i^{*b} \cap (C[z] \otimes C^r)$. 
 Then by the Hahn-Banach theorem there is a continuous linear functional that separates them, which is contradiction since all such functionals are given by the differential operator functionals.
 \end{proof}


\subsection{Shift-invariant subspaces}
The next two lemmas enable the transference of orthogonality and dual space density results between modules in $\bC[z]\otimes \bC^r$ and modules in $\bC(z)\otimes \bC^r$. 

\begin{lem}\label{l:BandBstar} The vectorial $pg$-sequence   $u_{\omega, \ul{h}}$ is orthogonal to the $\bC(z)$-module ${B}$ if and only if it is orthogonal to the $\bC[z]$-module ${B}^*$.
\end{lem}

\begin{proof}
Note that for fixed $p(z)= (p_1(z),\dots ,p_r(z))$ in $B^*$ and fixed $h(z)= (h_1(z),\dots ,h_r(z))$ in $\bC[z]\otimes \bC^r$ we have
\[
\langle z^ip(z), u_{\omega,h}\rangle = \sum_{t=1}^r \langle z^ip_t(z), (h_t(k)\omega^k)_k\rangle = \pi(i)\omega^i. 
\]
for some polynomial $\pi(z)$. This is clear if the polynomials $p_t, h_t$ are monomials and so it follows in general by linearity. If these terms are zero for all $i \in \bZ^d_+$ then $\pi(i)$ is zero for all such $i$ and so $\pi(z)$ is the zero polynomial, and hence the terms are equal to zero for all $i \in \bZ^d$.
Since $B$ is the union of the spaces $z^iB^*$, for all multi-indices $i$, the lemma follows.
\end{proof}

\begin{lem}\label{l:Aplusdensity}
Let $A$ be a closed invariant subspace of $C({\bZ^d})\otimes \bC^r$ and let
$A_+\subseteq C(\bZ^d_+)\otimes \bC^r$ be the set of restrictions of sequences $u$ in $A$. Also, let $\P$ be a  invariant linear space of vectorial $pg$-sequences in $A$ whose restrictions to $\bZ^d_+$ form a dense set in $A_+$. Then $\P$ is dense in $A$.
\end{lem}

\begin{proof} 
 Identify $A_+$ with the corresponding set of $\bZ^d$-sequences $(w_k)$ which are zero if $k\notin \bZ^d_+$. Similarly define $\P_+$. Since $A$ is shift-invariant,
each $u\in A$ is the limit of a sequence of elements of the form $(W_1\cdots W_d)^{-n}(u^n)_+$, with  $u^n\in A$. By the hypotheses,  each  $(u^n)_+$ is approximable by elements $w_+$ of $\P_+$ where $w$ a linear combination of $pg$-sequences in $A$. It follows that $u$ is also approximable by the corresponding sequence of elements $(W_1\cdots W_d)^{-n}w$ in $A$. Since these elements are linear combinations of $pg$-sequences in $A$ the lemma follows.
\end{proof}

\begin{thm}\label{t:vectorialsynthesis}
Let $A$ be a closed invariant subspace of $C({\bZ^d})\otimes \bC^r$.
Then there is a finite set of geometric indices 
such that $A$ is the closed linear span
of the vectorial $pg$-sequences in $A$ with geometric indices in this set.
\end{thm}

\begin{proof}Let $B$ be the annihilator of $A$ with associated $\bC[z]$-module $B^*$. By  (\ref{bigmoduledecomp}) we have the decomposition
\[
B^* = (Q_1^{*b}\cap (\bC[z]\otimes \bC^r))\cap \dots \cap (Q_s^{*b} \cap (\bC[z]\otimes \bC^r))
\]
associated with any choice of roots $\omega(1),\dots ,\omega(s)$ for the associated primary submodules $Q_i$.
By Lemma \ref{l:mainlemma} a vector polynomial $p(z)$ lies in $B^*$ if and only if for each $1 \leq i\leq s$  it is orthogonal to every vectorial $pg$-sequence $u_{\omega(i),h}$ which is orthogonal to $Q_i^*$.
It follows that the set of all the functionals $L$ in $(\bC[z]\otimes \bC^r)'$ of the form 
\[
L_{\omega(i),h}: p(z) \to \langle p, u_{\omega(i),h}\rangle, \quad h \in Q_i^{*b}, 1\leq i \leq s,
\]  determine membership in  $B^*$. That is, if $L(p(z))=0$ for all such $L$ 
with $L(Q_i^{*b})={0}$, for all $i\in{1,...,s}$ then $p(z)\in B^*$. By the reflexivity of $\bC[z]\otimes \bC^r$ it also follows that this specific set of functionals which annihilate $B^*$ has dense linear span in $(B^*)^\perp$. Let us write $\S_+$ for this subset and $\S$ for the set of corresponding functionals on $\bC(z)\otimes \bC^r$.

By  Lemma \ref{l:BandBstar} the set $\S$ consists of the differential operator functionals that annihilate $B$.  In particular $\S$  is an invariant set for the shift operators and their inverses. By Lemma \ref{l:Aplusdensity} it follows that the linear span of this set is dense in $A$, as desired.
\end{proof}

\begin{rem} 
We have followed the general proof scheme of Lefranc's succinct 1958 paper \cite{lef}. However our arguments also give fuller details when specialised  $r=1$ and ideals. 
The only other account of the proof that we are aware of is in
de Boor and Ron \cite{boo-ron} where applications are made to multivariate spline approximation.

More recently algebraic spectral synthesis has been examined for general discrete groups and is now known to hold for the discrete groups whose torsion-free rank is finite. See Laczkovich and Szekelyhidi \cite{lac-sze}
for further details.
\end{rem}

\begin{rem}
Proposition \ref{p:diffop}, in the scalar case, identifies the dual space of the power series ring as a space of differential operator functionals. 
In combination with the dual space identifications in Section \ref{ss:duality} and the Hahn-Banach theorem this identification shows that the differential operator functionals determine ideal membership.  The following separation theorem is a version of this 
for {constant coefficient} differential operator functionals. 

\begin{thm}\label{t:idealseparation} Let $\I$ be an ideal in the  ring $\bC[z]$ and let $p(z)$ be a polynomial in $\bC[z]$ which is not in $\I$. Then
there is a constant coefficient linear differential operator $D = \sum_{k\in \bZ^d} c_k\partial^{k}$ and $\omega = (\omega_1, \dots , \omega_d)\in \bC^d$ such that $Df(\omega) = 0$ for  $f \in I$ and $Dp(\omega) \neq 0$.
\end{thm}

\noindent This result and other applications of Lefranc's theorem are  discussed in Szekelyhidi \cite{sze-book}.
Also Laczkovich \cite{lac} has recently obtained an interesting generalisation of Theorem \ref{t:idealseparation}  for rings with countably many variables and differential operators which are infinite sums. 
\end{rem}

\section{The proofs of Theorem \ref{t:irigid}, Theorem \ref{t:flexthm} and Theorem \ref{t:findimflex}}

The following degree reduction lemma will be used in the proof of Theorem \ref{t:irigid}. We consider the multi-degrees $k$ of the monomials $z^k$ to be ordered according to the lexicographic ordering on $\bZ^d_+$.

\begin{lem}\label{l:degreereduce}
Let $u: k \to \omega^kh(k)$ be a nonzero vectorial $pg$-sequence in $\bC[z]\otimes \bC^r$ with
$h(z)=(h_1(z), \dots , h_r(z))$ where $h_i(z)$ has multi-degree $\delta(i) \in \bZ^d$ and let $A_0$ be the (unclosed) linear span of the $\bZ^d$-translates of $u$. If $|\delta(i)|\geq 2$ for some $i$ then there exists a nonzero  vectorial $pg$-sequence $w: k\to \omega^kg(k)$ in $A_0$ with $g(z)$ a nonconstant linear vector-valued polynomial in $\bC[z]\otimes \bC^r$.
\end{lem}

\begin{proof}
Let $p(z)\in \bC[z]$ with $z^k=z_1^{k_1}\cdots z_d^{k_d}$  the leading term of $p(z)$ in the lexicographic order. If $k_1\geq 2$ then $p(z)-p(z_1-1,z_2,\dots ,z_d)$ is a polynomial of lower multi-degree. The lemma follows by successively repeating such degree reduction. Specifically, suppose that $\delta\in \bZ^d_+$ is the largest multi-degree for the coordinate functions of $h_i(z), 1 \leq i \leq r$, possibly appearing for several values of $i$. This, by definition, is the multi-degree  of  $h(z)$ and is the maximum of the multidegrees $\delta(i), 1 \leq i \leq s$. Let $\delta_j$ be the first nonzero exponent for $\delta$. Then the vector-valued polynomial $h(z) - (W_j\otimes I_r)h(z)$ is in $A_0$ and has lower multi-degree.
\end{proof}

\begin{proof}[Proof of Theorem \ref{t:irigid}]
Note that (v) is equivalent to (i) by the duality assertions of Lemma 4.1. Also (i) evidently implies (ii). To see that (ii) implies (i) we must show that if there is a first-order flex which is not a rigid motion flex
then in fact there exists an exponential flex which is a nonrigid motion flex. This conclusion follows immediately from Theorem \ref{t:vectorialsynthesis} which shows that in fact there must exist a nonrigid motion flex  $u_{\omega, \ul{h}}$.

Assertions (iii) and (iv) are equivalent, by the discussion preceding Definition \ref{d:flaxlatticerigid}, and they are implied by (i). 

It remains to show that (iii) is a sufficient condition for (i). Assume the contrary, that (i) does not hold and (iii) holds. Once again, by  Theorem \ref{t:vectorialsynthesis}, there exists a nonrigid motion flex  $u= u_{\omega, \ul{h}}$. Suppose first that $\omega = \ul{1}$. Since (iii) holds there is no strictly periodic nonrigid motion flex and so not all of the coordinate polynomials $h_1(z), \dots , h_d(z)$ can be constant polynomials. 
If they are all linear or constant polynomials then $u$ is a flexible lattice periodic flex and this is a contradiction. 
However, in general
we may apply Lemma \ref{l:degreereduce} to reduce to this case and so once again obtain the desired contradiction.
Suppose, finally, that $\omega \neq \ul{1}$. Then by the proof of Lemma \ref{l:degreereduce} we may successively obtain flexes with reduced  multi-degrees to obtain a geometric flex of the form $u_{\omega}\otimes a$. This means that $\omega$ is in the geometric spectrum which is a contradiction.
\end{proof}

\begin{proof}[Proof of Theorem \ref{t:flexthm}]For a given periodic structure
the first order flex space $\F(\C;\bC)$ is the linear space dual of the $\bC(z)$-module $M(\bC)^*$ under the natural pairing, as in Section \ref{ss:duality}. Thus the theorem follows from Theorem \ref{t:vectorialsynthesis} and its proof.
\end{proof}

\begin{proof}[Proof of Theorem \ref{t:findimflex}]
If $\Gamma(\C)$ is infinite then the flex space is infinite-dimensional since a finite set of geometric flexes with distinct periodicity factors is linearly independent. On the other hand if $\Gamma(\C)$ is a finite set $\omega(1), \dots ,\omega(s)$ then the $\bC[z]$-module $M(\C)^*$ has a primary decomposition of length $s$ with primary module $Q_i^*$ having the unique root $\omega(i)$, for $1\leq i \leq s$. It follows that the annihilator of each $Q_i^*$ is  finite-dimensional and that the annihilator of $M(\C)^*$, being the closed span of these spaces, is finite-dimensional. By Theorem \ref{t:flexthm} this annihilator is equal to the first-order flex space of $\C$ and so the proof is complete. \end{proof}

\section{first-order rigid and yet continuously flexible}

We now consider direct geometric arguments to show that a crystallographic bar-joint framework may be continuously flexible even when it is first-order rigid. This phenomenon is not possible for finite bar-joint frameworks (Asimow and Roth \cite{asi-rot}) since one may use the algebraic variety structure of the configuration space to show that the existence of a continuous flex implies the existence of a differentiable flex.

Consider first the semi-infinite periodic strip framework $\Q_{\rm right}= (G,p)$ suggested by Figure \ref{f:InfRig} where
the triples of joints $\{A,X,Q\}, \{B,Y,S\}, \dots $ are  collinear. We claim that this is first-order rigid.

\begin{center}
\begin{figure}[ht]
\centering
\includegraphics[width=5cm]{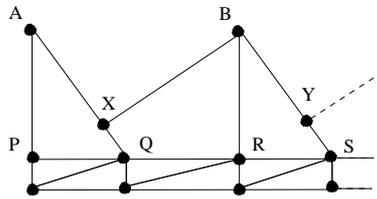}
\caption{The semi-infinite strip framework $\Q_{\rm right}$.}
\label{f:InfRig}
\end{figure}
\end{center}

To see this suppose, by way of contradiction, that there is a nonzero velocity field $u$ 
which assigns zero velocities to the
the joints lying on and below the line through $PQRS$. Let $u_X, u_Y, u_Z, \dots $ be the velocities for the joints $X,Y,Z, 
\dots $.  One of these velocities must be nonzero and without 
loss of generality we may assume  $u_X\neq 0$. Thus $u_B\neq 0$ and is in the direction of the 
positive $x$-axis. However, $u_S=0$ and $B,Y,S$ are collinear and so 
this is a contradiction since there is no finite  velocity $u_Y$ such that $u_B, u_Y, u_S$ satisfy the flex conditions for the edges $BY$ and $YS$. 

We next claim that for a suitable choice of geometry the framework is continuously flexible.

Assume first that  $|XQ| < |QR|$ and  $|QB| >|XB| >|RB|$. Consider the 
finite subframework,
$(G_1,\pi)$ say, supported by the labelled vertices and the four vertices below $P,Q,R,S$.
For this subframework, consider the joints $P,Q,R,S$ as fixed.
Let $AP$  rotate at constant speed
through a clockwise angle $t>0$, so that the bar $XQ$ (with infinite initial velocity)  rotates continuously clockwise to
achieve a horizontal position corresponding to the final clockwise angle $t=t_1$ say.
The induced angular positions $\theta(t)$ of $BR$ in this motion increase  
first to a local maximum, $\theta_{max}$, when $QX$ and $XB$ are  co-linear, and then decreases through positive values to a final value 
$\theta_{fin}=\theta(t_1)$. Assume now that, 
\[
|XB | \geq \sqrt{|BR|^2+(|RQ|-|RX|)^2}
\]
so that $\theta_{fin} >0$. It follows that the  range of the continuous function 
$\theta:t \to \theta(t),$ for $ 0\leq t \leq t_1,$ is included in the  range of its argument $t$.
Iterating this inclusion principle it follows that 
the continuous flex $\pi(t)$ of $(G_1,\pi)$, with flex parameter  $0\leq t\leq  
t_1$, is extendible to a continuous flex $t \to p(t), 0\leq t \leq t_1,$ of the framework $\Q_{\rm right}$.

The continuous flex $p(t)$, with full parameter range $0\leq t \leq t_1$, does not extend to a continuous flex of the two-sided periodic strip, 
$\Q$ say. Indeed the maximum possible positive angular deviation of any vertical bar of $\Q$, such as $BR$, is limited by the colinearity position of bars $QX$ and $XB$ to the \emph{left} of $BR$.
However, we claim that $\Q$ is continuously flexible, with the angular motions of all the vertical bars taking place within the range $0\leq\theta \leq \theta_{max}$. 

To see this consider again the angle propagation function $\theta : t \to \theta(t)$, defined for $0\leq t \leq t_1$. 
Let
$t= t_{\rm fix}$ be the positive solution of
$\theta(t)= t$ and note that $t_{\rm fix} < t_max$. As we have observed, the angular motion or position of the $n$-th vertical bar to the \emph{right} is governed by the
iterates of $\theta$ and it follows that as $t$ tends to $t_{\rm fix}$
the angular inclination of each vertical bar on the right converges to $t_{fix}$. 

We claim that the continuous motion of $(G_1,\pi)$ parametrised by $0\leq t \leq t_{fix}$, can be extended to the left strip of $\Q$ and hence defines a continuous flex of $\Q$.
To see this consider once more the finite subframework linking $AP$ and $BR$, but  with $BR$ providing the driving angular displacement and flex parameter $s\geq 0$. The leftward angle propagation function is the inverse function $ s\to \theta^{-1}(s)$. This is a smooth {decreasing} function, well defined for the range $0<s \leq \theta_{max}$, with derivative $0$ at $s=0$. It follows that the motion is extendible to the left hand side and the claim follows.

It is now straightforward to construct a crystal framework in $\bR^2$, which is continuously flexible and first-order rigid, by taking parallel copies of the strip framework $\Q$ and rigidly connecting their rigid base subframeworks in a periodic manner.

\begin{rem}
The continuous  flex $t \to p(t), 0\leq t \leq t_{\rm fix},$ of the strip framework $\Q$
 may be reparametrised in terms of the inclination angle interval $0\leq \gamma \leq \gamma_n$ of any fixed vertical bar. However, any such parametrisation fails to provide a smooth flex since the derivative at time zero with respect to $\gamma$ (the initial velocity) for any moving joint to the right of this vertical, is infinite.
\end{rem}

\subsection{Aperiodic phase transitions}
The continuous flex $t \to p(t)$  of the strip framework $\Q$ adopts aperiodic positions for each intermediate value of $t$, with  $0< t < t_{\rm fix}$, while the terminal position, for $t=t_{\rm fix}$, is a periodic strip framework which we denote as $\Q_1$.
Thus $\Q_1$  is a tilted placement of $\Q$ and we can view the motion as an \emph{aperiodic phase transition} between 2 periodic states. By varying the initial geometry, so that in the initial periodic position $A, X, Q$ are not collinear, one can also construct strip frameworks with {aperiodic phase transitions} which are smooth paths.
By embedding strip frameworks such as these in higher dimensional constructions one can obtain 3D periodic frameworks with similar aperiodic phase transitions between crystal states. It would be interesting to discover if such locally chaotic transitions between periodic states could serve as a model for abrupt transitions in material crystals, such as martensitic changes of state. See Anwar et al \cite{anw-et-al} for example.

\bibliographystyle{abbrv}
\def\lfhook#1{\setbox0=\hbox{#1}{\ooalign{\hidewidth
  \lower1.5ex\hbox{'}\hidewidth\crcr\unhbox0}}}

\end{document}